\documentclass[a4paper,UKenglish,cleveref, autoref]{lipics-v2021}

\hideLIPIcs  %uncomment to remove references to LIPIcs series (logo, DOI, ...), e.g. when preparing a pre-final version to be uploaded to arXiv or another public repository

\usepackage{graphicx}

\usepackage{algpseudocode}
\usepackage{algorithm} %to put captions in algo

\usepackage[dvipsnames]{xcolor}

\usepackage{tikz} 
\usetikzlibrary{arrows.meta}
\usetikzlibrary{chains,positioning,scopes,shapes} % drawing trees

\newcommand{\licol}{\textsc{Li $3$-col}}
\newcommand{\likcol}{\textsc{Li $k$-col}}
\newcommand{\calL}{\mathcal{L}}

\algnewcommand{\LeftComment}[1]{\State \(\triangleright\) #1}

\title{List 3-Coloring on Comb-Convex and Caterpillar-Convex Bipartite Graphs} %TODO Please add

\author{Banu Baklan \c{S}en}{Computer Engineering Department, Kadir Has University, Istanbul, Turkey}{banubak@hotmail.com}{https://orcid.org/0000-0003-4545-5044}{}%{(Optional) author-specific funding acknowledgements}

\author{\"{O}znur Ya\c{s}ar Diner}{Computer Engineering Department, Kadir Has University, Istanbul, Turkey}{oznur.yasar@khas.edu.tr}{https://orcid.org/0000-0002-9271-2691}{}

\author{Thomas Erlebach}{Department of Computer Science, Durham University, Durham, United Kingdom} {thomas.erlebach@durham.ac.uk}{https://orcid.org/0000-0002-4470-5868}{}

\authorrunning{B. Baklan \c{S}en, \"{O}.Y. Diner, T. Erlebach} %TODO mandatory. First: Use abbreviated first/middle names. Second (only in severe cases): Use first author plus 'et al.'

\Copyright{Banu Baklan \c{S}en and \"{O}znur Ya\c{s}ar Diner and Thomas Erlebach} %TODO mandatory, please use full first names. LIPIcs license is "CC-BY";  http://creativecommons.org/licenses/by/3.0/


 %TODO mandatory; please add comma-separated list of keywords

\category{} %optional, e.g. invited paper

\relatedversion{An extended abstract appears in the proceedings of the 29th International Computing and Combinatorics Conference (COCOON 2023).} %optional, e.g. full version hosted on arXiv, HAL, or other respository/website

\nolinenumbers %uncomment to disable line numbering

\EventEditors{...}
\EventNoEds{2}
\EventLongTitle{Long}
\EventShortTitle{XYZ23}
\EventAcronym{XYZ}
\EventYear{2023}
\EventDate{August}
\EventLocation{X}
\EventLogo{}
\SeriesVolume{..}
\ArticleNo{..}
\begin{document}

\maketitle

\begin{abstract}
Given a graph $G=(V, E)$ and a list of available colors $L(v)$ for each vertex $v\in V$, where $L(v) \subseteq \{1, 2, \ldots, k\}$, {\sc List $k$-Coloring} refers to the problem of assigning colors to the vertices of $G$ so that each vertex receives a color from its own list and no two neighboring vertices receive the same color. The decision version of the problem {\sc List $3$-Coloring} is NP-complete even for bipartite graphs, and its complexity on comb-convex bipartite graphs has been an open problem. We give a polynomial-time algorithm to solve {\sc List $3$-Coloring} for caterpillar-convex bipartite graphs, a superclass of comb-convex bipartite graphs. We also give a polynomial-time recognition algorithm for the class of caterpillar-convex bipartite graphs.
\end{abstract}

\section{Introduction}
\label{sec:introduction}

Graph coloring is the problem of assigning colors to the vertices of a given graph
in such a way that no two adjacent vertices have the same color.
\emph{List coloring} \cite{vizing,erdos} is a generalization of
graph coloring in which each vertex must receive a color from its own list of
allowed colors.
In this paper, we study the list coloring problem with a fixed number of colors in subclasses of bipartite graphs. We give a polynomial-time algorithm for the list $3$-coloring problem for caterpillar-convex bipartite graphs, a superclass of comb-convex bipartite graphs. We also give a polynomial-time recognition algorithm for the class of caterpillar-convex bipartite graphs. Our results resolve the open question regarding the complexity of
list $3$-coloring for comb-convex bipartite graphs stated in~\cite{flavia,brettell_arxiv}.

We consider finite simple undirected graphs $G=(V, E)$ with vertex set $V$ and edge set $E$.
By $N_G(v)$ (or by $N(v)$ if the graph is clear from the context) we denote the neighborhood of $v$ in $G$, i.e., the set of vertices that are adjacent to~$v$. A \emph{$k$-coloring} of $G$ is a labeling that assigns colors to the vertices of $G$ from the set $[k]=\{1, 2,  \ldots, k\}$. A coloring is \emph{proper} if no two adjacent vertices have the same color.
A \emph{list assignment} of a graph $G=(V, E)$ is a mapping $\mathcal{L}$ that assigns each vertex $v \in V$ a list $\mathcal{L}(v) \subseteq \{1, 2,\ldots \} $ of admissible colors for~$v$. When $\mathcal{L}(v) \subseteq [k]=\{1, 2,\ldots k\}$ for every $v \in V$ we say that $\mathcal{L}$ is a $k$-list assignment of $G$. The total number of available colors is bounded by $k$ in a $k$-list assignment. On the other hand, when the only restriction is that $|\mathcal{L}(v)| \leq k$ for every $v \in V$, then we say that $\mathcal{L}$ is a list $k$-assignment of $G$.
\emph{List coloring} is the problem of deciding, for a given graph $G=(V,E)$
and list assignment $\mathcal{L}$, whether $G$ has a proper coloring
where each vertex~$v$ receives a color from its list $\mathcal{L}(v)$.
If $\calL$ is a $k$-list assignment for a fixed value of~$k$,
the problem becomes the list $k$-coloring problem:

\smallskip
\noindent
\underline{{\sc List $k$-Coloring} ({\sc Li $k$-Col})}\\
{\it Instance:} \hspace*{0.2mm} A graph $G=(V,E)$ and a $k$-list assignment $\calL$.\\
{\it Question:} Does $G$ have a proper coloring where each vertex $v$ receives a color from its list $\calL(v)$? 
\smallskip

If $\mathcal{L}$ is a list $k$-assignment instead of a $k$-list assignment, the
problem is called {\sc $k$-List Coloring}.

\begin{figure}[t] %%%%%%%%% FIG 1 (alone)
\begin{center}
\resizebox{0.9\textwidth}{!}{%
\begin{tikzpicture}[dot/.style={draw,circle,minimum size=1mm,inner sep=0pt,outer sep=0pt}]

\draw (-2,8.25) node {$x_1$};
\draw (-2,7.25) node {$x_2$};
\draw (-2,6.25) node {$x_3$};
\draw (-2,5.25) node {$x_4$};
\draw (-2,4.25) node {$x_5$};
\draw (-2,3.25) node {$x_6$};

\draw (1,7.25) node {$y_1$};
\draw (1,6.25) node {$y_2$};
\draw (1,5.25) node {$y_3$};
\draw (1,4.25) node {$y_4$};

\coordinate [dot,fill = black] (X6) at (-2,3);
\coordinate [dot,fill = black] (X5) at (-2,4);
\coordinate [dot,fill = black] (X4) at (-2,5);
\coordinate  [dot,fill = black](X3) at (-2,6);
\coordinate  [dot,fill = black](X2) at (-2,7);
\coordinate  [dot,fill = black](X1) at (-2,8);

\coordinate[dot,fill = black] (Y4) at (1,4);
\coordinate [dot,fill = black] (Y3) at (1,5);
\coordinate [dot,fill = black] (Y2) at (1,6);
\coordinate [dot,fill = black] (Y1) at (1,7);

\coordinate [dot,fill = black] (X7) at (2,5);
\coordinate [dot,fill = black] (X8) at (2,6);
\coordinate [dot,fill = black] (X9) at (2,7);

\coordinate [dot,fill = black] (Y5) at (3,5);
\coordinate [dot,fill = black] (Y6) at (3,6);
\coordinate [dot,fill = black] (Y7) at (3,7);

\draw (1.75,7) node {$x_1$};
\draw (1.75,6) node {$x_3$};
\draw (1.75,5) node {$x_5$};

\draw (3.25,7) node {$x_2$};
\draw (3.25,6) node {$x_4$};
\draw (3.25,5) node {$x_6$};

\draw (-0.25,2.5) node {(a)};

\draw (2.55,2.5) node {(b)};

\draw (6.55,2.5) node {(c)};

\draw (9.75,2.5) node {(d)};
\draw [black] (X9) -- (Y7);
\draw [black] (X8) -- (Y6);
\draw [black] (X7) -- (Y5);
\draw [black] (X9) -- (X8);
\draw [black] (X8) -- (X7);

\draw [black] (X1) -- (Y1) -- (X2);
\draw [black] (Y1) -- (X3);
\draw [black] (X1) -- (Y2) -- (X3); 
\draw [black] (Y2) -- (X4);
\draw [black] (X3) -- (Y3) -- (X5);
\draw [black] (X1) -- (Y3);
\draw [black] (X3) -- (Y4) -- (X5);
\draw [black] (Y4) -- (X6);

\draw (5,8.7) node {$x_1$};
\draw (5,7.95) node {$x_2$};
\draw (5,7.2) node {$x_3$};
\draw (5,6.45) node {$x_4$};
\draw (5,5.7) node {$x_5$};
\draw (5,4.95) node {$x_6$};
\draw (5,4.22) node {$x_7$};
\draw (5,3.45) node {$x_8$};

\draw (8,8.25) node {$y_1$};
\draw (8,7.25) node {$y_2$};
\draw (8,6.25) node {$y_3$};
\draw (8,5.25) node {$y_4$};
\draw (8,4.25) node {$y_5$};

\coordinate [dot,fill = black] (X8) at (5,3.25);
\coordinate [dot,fill = black] (X7) at (5,4);
\coordinate [dot,fill = black] (X6) at (5,4.75);
\coordinate [dot,fill = black] (X5) at (5,5.5);
\coordinate [dot,fill = black] (X4) at (5,6.25);
\coordinate  [dot,fill = black](X3) at (5,7);
\coordinate  [dot,fill = black](X2) at (5,7.75);
\coordinate  [dot,fill = black](X1) at (5,8.5);

\coordinate[dot,fill = black] (Y5) at (8,4);
\coordinate[dot,fill = black] (Y4) at (8,5);
\coordinate [dot,fill = black] (Y3) at (8,6);
\coordinate [dot,fill = black] (Y2) at (8,7);
\coordinate [dot,fill = black] (Y1) at (8,8);

\coordinate [dot,fill = black] (X9) at (9,5);
\coordinate [dot,fill = black] (X10) at (9,6);
\coordinate [dot,fill = black] (X11) at (9,7);

\coordinate [dot,fill = black] (Y6) at (10,4.5);
\coordinate [dot,fill = black] (Y7) at (10,5.5);
\coordinate [dot,fill = black] (Y8) at (10,6);
\coordinate [dot,fill = black] (Y9) at (10,6.5);
\coordinate [dot,fill = black] (Y10) at (10,7.5);

\draw (8.75,7) node {$x_1$};
\draw (8.75,6) node {$x_3$};
\draw (8.75,5) node {$x_5$};

\draw (10.25,7.55) node {$x_4$};
\draw (10.25,6.5) node {$x_2$};
\draw (10.25,6) node {$x_7$};
\draw (10.25,5.55) node {$x_6$};
\draw (10.25,4.55) node {$x_8$};

\draw [black] (X11) -- (Y10);
\draw [black] (X11) -- (Y9);
\draw [black] (X11) -- (X10);
\draw [black] (X10) -- (Y8);
\draw [black] (X10) -- (X9);
\draw [black] (X9) -- (Y7);
\draw [black] (X9) -- (Y6);

\draw [black] (X3) -- (Y1) -- (X7);
\draw [black] (X1) -- (Y2) -- (X3);
\draw [black] (X1) -- (Y3) -- (X3);
\draw [black] (X2) -- (Y3) -- (X4);
\draw [black] (X3) -- (Y4) -- (X1);
\draw [black] (X1) -- (Y4) -- (X5);
\draw [black] (X3) -- (Y5) -- (X1);
\draw [black] (X5) -- (Y5) -- (X6);
\draw [black] (X6) -- (Y5) -- (X8);

\end{tikzpicture}}
\end{center}

 \caption{(a) A comb-convex bipartite graph $G_1$, (b) a comb representation of $G_1$, (c) a caterpillar-convex bipartite graph $G_2$, and (d) a caterpillar representation for $G_2$.} \label{fig:comb}
\end{figure}
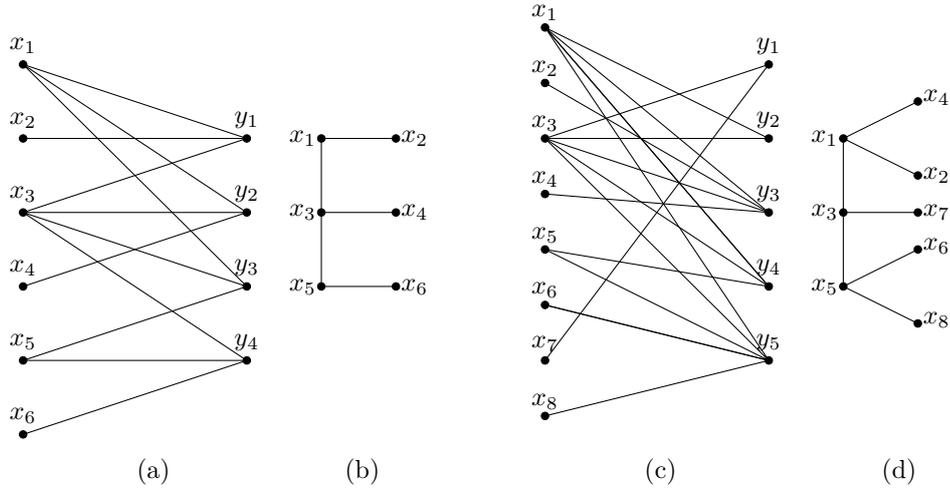%
The classes of bipartite graphs
of interest to us are defined via a convexity
condition for the neighborhoods of the vertices on one side of the graph with respect to a tree defined
on the vertices of the other side. The following types of trees are relevant here:
A \emph{star} is a tree of diameter at most~$2$.
A \emph{comb} is a tree that consists of a chordless
path $P$, called the \emph{backbone}, with a single leaf neighbor
attached to each backbone vertex~\cite{chen}.
A \emph{caterpillar} is a tree that consists of a chordless path~$P$,
called the \emph{backbone}, with an arbitrary number (possibly zero) of leaf vertices attached to each vertex on~$P$.
Note that if a caterpillar has exactly one leaf vertex attached to each vertex on~$P$, then
that caterpillar is a comb.

A bipartite graph $G=(X\cup Y,E)$ is called a \emph{star-convex} (or \emph{comb-convex}, or \emph{caterpillar-convex}) bipartite graph if a star (or comb, or caterpillar) $T=(X,F)$  can be defined on $X$ such that for each vertex $y \in Y$, its neighborhood $N_G(y)$ induces a subtree of $T$.
The star (or comb, or caterpillar) $T=(X,F)$ is then called a
\emph{star representation} (or \emph{comb representation}, or \emph{caterpillar representation}) of~$G$.

Figure~\ref{fig:comb} shows an example of a comb-convex bipartite graph and its comb representation, and a caterpillar-convex bipartite graph and its caterpillar representation. Both the comb
(in part (b)) and the caterpillar (in part (d)) have the
path $P=x_1 x_3 x_5$ as backbone.

The remainder of this paper is organized as follows. In Section 2, we discuss related work. In Section 3, we give a polynomial-time algorithm for \licol\ for caterpillar-convex bipartite graphs (and thus also for comb-convex bipartite graphs). In Section 4, we give a polynomial-time recognition algorithm for caterpillar-convex bipartite graphs. In Section 5, we give concluding remarks.

\section{Related Work}

Deciding whether a graph has a proper coloring with $k$ colors is
polynomial-time solvable when $k=1$ or $2$ \cite{lovasz} and NP-complete for $k\geq3$ \cite {vizing}. 
As \likcol\ generalizes this problem, it is also NP-complete for $k\geq 3$. 
When the list coloring problem is restricted to perfect graphs and their subclasses, it is still NP-complete in
many cases such as for bipartite graphs \cite{kubale} and interval graphs~\cite{biro}.
On the other hand, it is polynomially solvable for trees
and graphs of bounded treewidth \cite{jansen}.
The problems {\sc Li $k$-col} and {\sc $k$-List Coloring} are polynomial-time solvable if $k\le 2$ and NP-complete if $k\ge 3$ \cite{lovasz,vizing}.
{\sc $k$–list coloring} has been shown to be NP-complete for small values of $k$ for complete bipartite graphs and cographs by Jansen and Scheffler \cite{jansen}, as observed in \cite{golovach}.
The {\sc $3$-List Coloring} problem is NP-complete even if each color occurs in the lists of at most
three vertices in planar graphs with maximum degree three, as shown by Kratochvil and Tuza~\cite{jkratochvil}.

We use the following standard notation for specific graphs:
$P_t$ denotes a path with $t$ vertices;
$K_t$ denotes a clique with $t$ vertices;
$K_{\ell,r}$ denotes a complete bipartite subgraph with parts of sizes $\ell$ and $r$;
$K_{1,s}^1$ denotes the $1$-subdivision of $K_{1,s}$ (i.e., every
edge $e=\{u,v\}$ of $K_{1,s}$ is replaced by two edges $\{u,w_e\}$ and $\{w_e,v\}$,
where $w_e$ is a new vertex);
and $sP_1 + P_5$ is the disjoint union of $s$ isolated vertices and a~$P_5$.
\likcol\ is known to be NP-complete even for $k = 3$ within the class of $3$–regular planar bipartite graphs \cite{jkratochvil3}. 
On the positive side, for fixed $k \geq 3$, {\sc Li $k$-col} is polynomially solvable for $P_5$-free graphs \cite{hong}. {\sc Li $3$-col} is polynomial for $P_6$-free graphs \cite{broersma} and for $P_7$-free graphs \cite{bonomo}.
{\sc Li $3$-col} is polynomial-time solvable for $(K^1_{1,s}, P_t)$-free graphs for every $s \geq 1$ and $t \geq 1$ \cite{chudnovsky}. {\sc Li $k$-col} is polynomial-time solvable for $(sP_1 + P_5)$-free graphs, which was proven for $s = 0$ by Ho{\`{a}}ng et al.~\cite{hong} and for every $s \geq 1$ by Couturier et al.~\cite{couturier}. 

 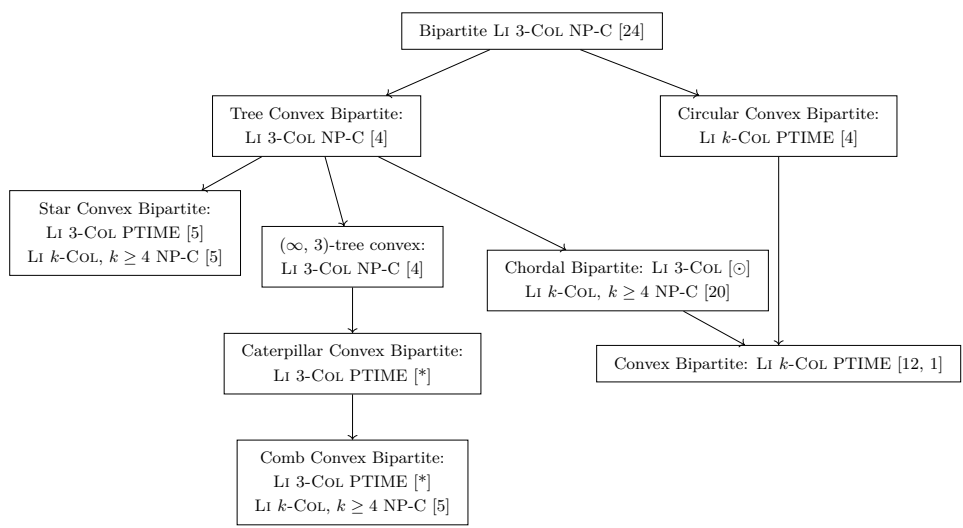
\begin{figure}[bt]
  \begin{center}
  \scalebox{0.7}{
  \begin{tikzpicture}[scale=0.9]

  \node[draw] (CxB) at (7.15,15) {
  \begin{tabular}{c} Convex Bipartite: {\sc Li $k$-Col} PTIME \cite{josep,belmonte} \end{tabular}};

    \node[draw] (ChB) at (4,16.75) {
  \begin{tabular}{c} Chordal Bipartite: {\sc Li $3$-Col} [$\odot$]\\{\sc Li $k$-Col, $k\geq 4$} NP-C \cite{huang} \end{tabular}};
  
  \node[draw] (CCxB) at (7.15,20) {
  \begin{tabular}{c} Circular Convex Bipartite: \\{\sc Li $k$-Col} PTIME \cite{flavia} \end{tabular}};
  
  \node[draw] (PB) at (2,22) {
  \begin{tabular}{c} Bipartite {\sc Li $3$-Col} NP-C \cite{kubale} \end{tabular}};

  \node[draw] (TCB) at (-2.5,20) {
  \begin{tabular}{c} Tree Convex Bipartite: \\{\sc Li $3$-Col} NP-C \cite{flavia} \end{tabular}};

  \node[draw] (SCB) at (-6.5,17.75) {
  \begin{tabular}{c} Star Convex Bipartite: \\ {\sc Li $3$-Col} PTIME \cite{brettell_arxiv}\\ {\sc Li $k$-Col}, $k\geq 4$ NP-C \cite{brettell_arxiv}\end{tabular}};

 \node[draw] (ITC) at (-1.75,17.25) {
  \begin{tabular}{c} ($\infty$, 3)-tree convex: \\ {\sc Li $3$-Col} NP-C \cite{flavia} \end{tabular}};

  \node[draw] (CCB) at (-1.75,15) {
  \begin{tabular}{c} Caterpillar Convex Bipartite: \\{\sc Li $3$-Col} PTIME [*]\\ \end{tabular}};

   \node[draw] (COB) at (-1.75,12.5) {
  \begin{tabular}{c} Comb Convex Bipartite: \\{\sc Li $3$-Col} PTIME [*]\\ {\sc Li $k$-Col}, $k\geq 4$ NP-C \cite{brettell_arxiv}\end{tabular}};

  \path (PB) edge[->] (TCB);
  \path (ITC) edge[->] (CCB);
  \path (TCB) edge[->] (ITC);
  \path (TCB) edge[->] (ChB);
  \path (CCB) edge[->] (COB);
  \path (CCxB) edge[->] (CxB);
  \path (TCB) edge[->] (SCB);
   \path (CxB) edge[<-] (ChB);
  \path (CCxB) edge[<-] (PB);
  
  \end{tikzpicture}
  }
  
 \end{center}
  
 \caption{Computational complexity results for {\sc Li $k$-Col} on subclasses of bipartite graphs. [*] refers to this paper, [$\odot$] refers to an open problem, NP-C denotes NP-complete problem, PTIME denotes polynomial-time solvable problem. }
 \label{fig:cat}
 \end{figure}
An overview of complexity results for \likcol\ in some subclasses of bipartite graphs is shown in Fig.~\ref{fig:cat}.
The computational complexity of {\sc Li $3$-col} for chordal bipartite graphs has been stated as an open problem in $2015$ \cite{huang} and has been of interest since then~\cite{josep}. In \cite{josep} a partial answer is given to this question by showing that {\sc Li $3$-col} is polynomial in the class of biconvex bipartite graphs and convex bipartite graphs. {\sc Li $3$-col} is solvable in polynomial time when it is restricted to graphs with all connected induced subgraphs having a multichain ordering  \cite{enright}. This result can be applied to permutation graphs and interval graphs. In \cite{josep}, it is shown that connected biconvex bipartite graphs have a multichain ordering, implying a polynomial time algorithm for {\sc Li $3$-col} on this graph class.
They also provide a dynamic programming algorithm to solve {\sc Li $3$-col} in the class of convex bipartite graphs and show how to modify the algorithm to solve the more general {\sc Li $H$-col} problem on convex bipartite graphs.
The computational complexity of {\sc Li $3$-col} for $P_8$-free bipartite graphs is open \cite{biro}. Even the restricted case of {\sc Li $3$-col} for $P_8$-free chordal bipartite graphs is open. Golovach et al.~\cite{golovach} survey
results for \likcol\ on $H$-free graphs in terms of the structure of~$H$.

So-called \emph{width parameters} play a crucial role in algorithmic complexity.
For various combinatorial problems, it is possible to find a polynomial-time solution by exploiting bounded width parameters such as mim-width, sim-width and clique-width.
Given a graph class, it is known that when mim-width is bounded, then {\sc Li $k$-col} is polynomial-time solvable~\cite{brettell}. Brettell et al.~\cite{brettell} proved that for every $r \geq 1, s \geq 1$ and $t \geq 1$, the mim-width is bounded and quickly computable for $(K_r, K^1_{1,s}, P_t)$-free graphs. This result further implies that for every $k \geq 1, s \geq 1$ and $t \geq 1$, {\sc Li $k$-col} is polynomial-time solvable for $(K^1_{1,s}, P_t)$-free graphs. 
Most recently, Bonomo-Braberman et al.~\cite{flavia} showed that mim-width is unbounded for star-convex and comb-convex bipartite graphs. On the other hand, {\sc Li $3$-col} is polynomial-time solvable for star-convex bipartite graphs whereas {\sc Li $k$-col} is NP-complete for $k\geq 4$~\cite{brettell_arxiv}. 
Furthermore, Bonomo-Braberman et al.~\cite{brettell_arxiv} show that for comb-convex bipartite graphs, {\sc Li $k$-col} remains NP-complete for $k\geq 4$ and leave open the computational complexity of {\sc Li $3$-col} for this graph class. In this paper, we resolve this problem by showing that {\sc Li $3$-col} is polynomial-time solvable even for caterpillar-convex bipartite graphs.

As for the recognition of graph classes, Bonomo-Braberman et al.~\cite{brettell_arxiv} provide an algorithm for the recognition of $(t, \Delta)$-tree convex bipartite graphs by using a result from Buchin et al.~\cite{buchin}. Here, a tree is a $(t,\Delta)$-tree if the maximum degree is bounded by $\Delta$ and the tree contains at most $t$ vertices of degree at least~$3$. This result for the recognition of $(t, \Delta)$-tree convex bipartite graphs, however, does not apply to caterpillar-convex bipartite graphs. Therefore, we give a novel algorithm for the recognition of caterpillar-convex bipartite graphs.

\section{List 3-Coloring Caterpillar-Convex Bipartite Graphs}
\label{sec:li3col}
In this section we give a polynomial-time algorithm for solving
\licol\ in caterpillar-convex bipartite graphs.
Let a caterpillar-convex bipartite graph $G=(X\cup Y,E)$ be given,
together with a 3-list assignment~$\calL$.
We assume that a caterpillar $T=(X,F)$ is also given, where
$N_G(y)$ induces a subtree of $T$ for each $y\in Y$.
If the caterpillar is not provided as part of the input,
we can compute one in polynomial time using the recognition algorithm that
we present in Section~\ref{sec:recognition}.

Let $T$ consist of a backbone $B$ with vertices
$b_1,b_2,\ldots,b_n$ (in that order) and a set of
leaves $L(b_i)$, possibly empty, attached to each
$b_i\in B$. We use $L$ to denote the set of
all leaves, i.e., $L=\bigcup_{i=1}^n L(b_i)$.
Furthermore, for any $1\le i\le j \le n$,
we let $B_{i,j}=\{b_i,b_{i+1},\ldots,b_j\}$
and $L_{i,j}=\bigcup_{k=i}^j L(b_k)$.

The idea of the algorithm is to define suitable subproblems
that can be solved in polynomial time, and to obtain the overall
coloring as a combination of solutions to subproblems.
Roughly speaking, the subproblems consider stretches of the
backbone in which all backbone vertices are assumed to be
assigned the same color in a proper list 3-coloring.
More precisely, a subproblem $SP(i,j,c_1,c_2,c_3)$ is specified via
two values $i,j$ with $1\le i\le j\le n$ and three
colors $c_1,c_2,c_3$ with $c_1\neq c_2$ and $c_2\neq c_3$ where $c_i \in [3],$ for $ i=1, 2, 3$.
Hence, there are $O(n^2)$ subproblems.

The subproblem $S=SP(i,j,c_1,c_2,c_3)$
is concerned with the subgraph $G_S$ of $G$ induced
by $B_{i-1,j+1}\cup L_{i,j} \cup \{y\in Y\mid N(y)\cap(B_{i,j}\cup L_{i,j})\neq\emptyset\}$.
It assumes that color $c_1$ is assigned to $b_{i-1}$,
color $c_2$ to the backbone vertices from $b_i$ to $b_j$,
and color $c_3$ to $b_{j+1}$. See Fig.~\ref{fig:SP} for
an illustration of $SP(i,j,2,1,2)$. Solving the subproblem $S$
means determining whether this coloring of $B_{i-1,j+1}$
can be extended to a proper list 3-coloring of $G_S$.
The result of the subproblem is False if this is not possible,
or True (and a suitable proper list 3-coloring of $G_S$) otherwise.
If $c_1\notin \calL(b_{i-1})$, or $c_3\notin \calL(b_{j+1})$,
or $c_2\notin \calL(b_k)$ for some $i\le k\le j$, then the
result of the subproblem is trivially False.

We will show that this subproblem can be solved in polynomial time
as it can be reduced to the 2-list coloring problem, which is
known to be solvable in linear time~\cite{Edwards86,kobler}.
Furthermore, solutions to consecutive `compatible' subproblems
can be combined, and a proper list 3-coloring of $G$ exists
if and only if there is a collection of subproblems whose
solutions can be combined into a list 3-coloring of~$G$.
For example, the colorings of two subproblems
$SP(i,j,c_1,c_2,c_3)$ and $SP(j+1,k,c_2,c_3,c_4)$ can be combined
because they agree on the colors of backbone vertices that are
in both subproblems, they do not share any leaf vertices, and
the vertices $y\in Y$ that have neighbors in both
$B_{i,j}\cup L_{i,j}$ and $B_{j+1,k}\cup L_{j+1,k}$
must be adjacent to $b_{j}$ and $b_{j+1}$,
which are colored with colors $c_2$ and $c_3$ (where $c_2\neq c_3$)
in the colorings of both subproblems, and hence must have received
the same color (the only color in $\{1,2,3\}\setminus\{c_2,c_3\}$)
in both colorings. To check whether there is a collection of
compatible subproblems whose solutions can be combined into a list
3-coloring of $G$, we will show that it suffices to search for
a directed path between two vertices in an auxiliary directed
acyclic graph (DAG) on the subproblems whose result is
True.

For a subproblem $S=SP(i,j,c_1,c_2,c_3)$,
if $i=1$, there is no vertex $b_{i-1}$, and we write
$\ast$ for $c_1$; similarly, if $j=n$, there is no
vertex $b_{j+1}$, and we write $\ast$ for $c_3$.
The graph $G_S$ considered when solving such a subproblem
does not contain $b_{i-1}$ or $b_{j+1}$, respectively,
but is otherwise defined analogously.
If $i=1$ and $j=n$, then $G_S$~contains neither
$b_{i-1}$ nor $b_{j+1}$.

\begin{lemma}
\label{lem:SP}
There is a linear-time algorithm for solving
any subproblem of the form $SP(i,j,c_1,c_2,c_3)$.
\end{lemma}

\begin{proof}
Consider the subproblem $S=SP(i,j,c_1,c_2,c_3)$.
Let $G_S$ be the subgraph of $G$ defined by~$S$,
and let $X_S\subseteq X$, $Y_S\subseteq Y$ be such
that the vertex set of $G_S$ is $X_S\cup Y_S$.
First, we check whether
$c_1\in \calL(b_{i-1})$ (only if $i>1$),
$c_3\in \calL(b_{j+1})$ (only if $j<n$),
and $c_2\in \calL(b_k)$ for all $i\le k\le j$.
If one of these checks fails, we return False.
Otherwise, we assign color $c_1$ to $b_{i-1}$,
color $c_2$ to all vertices in $B_{i,j}$, and
color $c_3$ to $b_{j+1}$.

For every vertex $y\in Y_S$, we check if $N(y)$
contains any vertices of $B_{i-1,j+1}$ and, if so,
remove the colors of those vertices from $\calL(y)$ (if they
were contained in $\calL(y)$).
If the list of any vertex $y\in Y_S$ becomes empty
in this process, we return False.

Let $B_S$ denote the backbone vertices in $X_S$
and $L_S$ the leaf vertices in $X_S$ (with respect
to the caterpillar $T$).
If there is a vertex in $L_S$ or $Y_S$ with a list
of size~$1$, assign the color in that list to that
vertex and remove that color from the lists of its
neighbors (if it is contained in their lists).
Repeat this operation until there is no uncolored
vertex with a list of size~$1$. (If an uncolored
vertex with a list of size~$1$ is created later
on in the algorithm, the same operation is applied
to that vertex.)
If the list of any vertex becomes empty in this
process, return False.
Otherwise, we must arrive at a state where
all uncolored vertices in $G_S$ have lists
of size 2 or~3.

If there is a vertex $y\in Y_S$ with a list of
size~3, that vertex must be adjacent to a single
leaf $\ell$ in $L_S$ (as it cannot be adjacent to a backbone
vertex). In this case we remove an arbitrary color from $\calL(y)$:
This is admissible as, no matter what color
$\ell$ receives in a coloring, vertex $y$ can always
be colored with one of the two colors that have remained in its list.

If there is a vertex $\ell \in L_S$ with a list
of length $3$, assign color $c_2$ to $\ell$
(and remove color $c_2$ from the lists of vertices
in $N(\ell)$). This color assignment does not
affect the existence of a proper list 3-coloring
for the following reasons (where we let $b_k$ denote the
backbone vertex with $\ell\in L(b_k)$):
\begin{itemize}
    \item If a vertex $y\in N(\ell)$ is adjacent to more
    than one vertex, it must be adjacent to $b_k$, which
    has been colored with~$c_2$, and hence it cannot receive color $c_2$
    in any case.
    \item If a vertex $y\in N(\ell)$ is adjacent only
    to $\ell$ and no other vertex, then $y$ can still
    be colored after $\ell$ is assigned color $c_2$,
    because we cannot have $\calL(y)=\{c_2\}$;
    this is because, 
    if $y$ had the list $\calL(y)=\{c_2\}$,
    it would have been colored $c_2$ and the color
    $c_2$ would have been removed from $\calL(\ell)$.
\end{itemize}
If at any step of this process, an uncolored vertex
with an empty list is created, return False. Otherwise,
we arrive at an instance $I$ of \licol\ where all uncolored
vertices have lists of size~$2$. Such an instance
can be solved in linear time~\cite{Edwards86,kobler} (via reduction to
a 2SAT problem).
If $I$ admits a proper list 3-coloring, that coloring
gives a proper list 3-coloring of $G_S$, and we
return True and that coloring. Otherwise, we return
False.

Correctness of the algorithm follows from its description,
and the algorithm can be implemented to run in linear
time using standard techniques.
\end{proof}

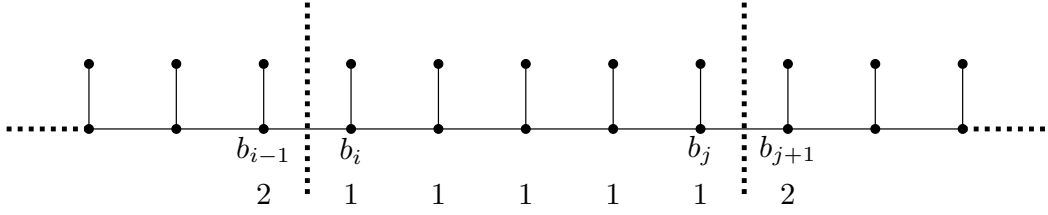
\begin{figure}[t] %%%%%%%%% FIG 1 (alone)
\begin{center}
\resizebox{0.99\textwidth}{!}{%
\begin{tikzpicture}[dot/.style={draw,circle,minimum size=1mm,inner sep=0pt,outer sep=0pt}]

\coordinate [dot,fill = black] (X1) at (-0.25,3);
\coordinate [dot,fill = black] (X2) at (0.75,3);
\coordinate [dot,fill = black] (X3) at (1.75,3);
\coordinate  [dot,fill = black](X4) at (2.75,3);
\coordinate  [dot,fill = black](X5) at (3.75,3);
\coordinate  [dot,fill = black](X6) at (4.75,3);
\coordinate  [dot,fill = black](X7) at (5.75,3);

\coordinate  [dot,fill = black](X8) at (6.75,3);
\coordinate  [dot,fill = black](X9) at (7.75,3);
\coordinate  [dot,fill = black](X10) at (8.75,3);
\coordinate  [dot,fill = black](X11) at (9.75,3);

\coordinate[dot,fill = black] (Y1) at (-0.25,3.75);
\coordinate [dot,fill = black] (Y2) at (0.75,3.75);
\coordinate [dot,fill = black] (Y3) at (1.75,3.75);
\coordinate [dot,fill = black] (Y4) at (2.75,3.75);
\coordinate [dot,fill = black] (Y5) at (3.75,3.75);
\coordinate [dot,fill = black] (Y6) at (4.75,3.75);
\coordinate [dot,fill = black] (Y7) at (5.75,3.75);
\coordinate [dot,fill = black] (Y8) at (6.75,3.75);
\coordinate [dot,fill = black] (Y9) at (7.75,3.75);
\coordinate [dot,fill = black] (Y10) at (8.75,3.75);
\coordinate [dot,fill = black] (Y11) at (9.75,3.75);

\draw [black] (X1) -- (Y1);
\draw [black] (X1) -- (X2);
\draw [black] (X2) -- (Y2);
\draw [black] (X2) -- (X3);
\draw [black] (X3) -- (Y3);
\draw [black] (X3) -- (X4);
\draw [black] (X4) -- (Y4);
\draw [black] (X4) -- (X5);
\draw [black] (X5) -- (Y5);
\draw [black] (X5) -- (X6);
\draw [black] (X6) -- (Y6);
\draw [black] (X6) -- (X7);
\draw [black] (X7) -- (Y7);
\draw [black] (X7) -- (X8);
\draw [black] (X8) -- (Y8);
\draw [black] (X8) -- (X9);
\draw [black] (X9) -- (Y9);
\draw [black] (X9) -- (X10);
\draw [black] (X10) -- (Y10);
\draw [black] (X10) -- (X11);
\draw [black] (X11) -- (Y11);

\draw[dotted] [ultra thick](-0.25,3) -- (-1.25,3);
\draw[dotted] [ultra thick](9.75,3) -- (10.75,3);

\draw[dotted] [ultra thick](2.25,2.25) -- (2.25,4.5);
\draw[dotted] [ultra thick](7.25,2.25) -- (7.25,4.5);

\draw (1.75,2.75) node {$b_{i-1}$};
\draw (1.75,2.25) node {$2$};
\draw (2.75,2.75) node {$b_i$};
\draw (2.75,2.25) node {$1$};
\draw (3.75,2.25) node {$1$};
\draw (4.75,2.25) node {$1$};
\draw (5.75,2.25) node {$1$};
\draw (6.75,2.25) node {$1$};
\draw (6.75,2.75) node {$b_j$};

\draw (7.75,2.75) node {$b_{j+1}$};
\draw (7.75,2.25) node {$2$};

\end{tikzpicture}}
\end{center}

 \caption{Illustration of subproblem $SP(i,j,c_1,c_2,c_3)$ for the case $SP(i, j, 2, 1, 2)$.} 
 \label{fig:SP}
\end{figure}

\begin{algorithm}[tbp]
\caption{List-3-Coloring Algorithm for Caterpillar-Convex Bipartite Graphs}
\label{algo1}
\begin{algorithmic}[1]
\Require A caterpillar-convex bipartite graph $G=(X \cup Y, E)$ (with caterpillar $T=(X,F)$) and a list assignment $\mathcal{L}$.
\Ensure A proper coloring that obeys $\mathcal{L}$, or False if no such coloring exists.
\LeftComment{Compute solutions to all subproblems}
\For {$i=1$ to $n$}
    \For {$j=i$ to $n$}
        \For {$c_i\in [3], i=1, 2, 3$ 
        with $c_1 \neq c_2$ and $c_2\neq c_3$}
        \State \(\triangleright\) let $c_1=\ast$ if $i=1$ and $c_3=\ast$ if $j=n$
        \State Solve $SP(i, j, c_1, c_2, c_3)$ (Lemma~\ref{lem:SP})
        \EndFor
    \EndFor
\EndFor
\LeftComment{Check if solutions of subproblems can be combined
into a list 3-coloring of $G$}
\State Build a DAG $H$ whose vertices are $s$, $t$, and
a vertex $SP(i,j,c_1,c_2,c_3)$ for each subproblem with
answer True (Definition~\ref{def:DAG})
\If{$H$ contains a directed path $P$ from $s$ to $t$}
\State Return the coloring obtained as union of the
colorings of the subproblems on~$P$
\Else
\State Return False
\EndIf
\end{algorithmic}
\end{algorithm}

Call a subproblem $S=SP(i,j,c_1,c_2,c_3)$ \emph{valid} if its
answer is True (and a proper list 3-coloring of $G_S$ has been produced), and \emph{invalid}
otherwise. To check whether the colorings obtained from valid subproblems
can be combined into a list 3-coloring of $G$, we make use of an auxiliary
DAG $H$ constructed as follows. The existence of a proper list 3-coloring of $G$ can then be
determined by checking whether $H$ contains a directed path
from $s$ to~$t$.

\begin{definition}
\label{def:DAG}
The auxiliary DAG $H=(V_H,A)$ has vertices $s$, $t$, and
a vertex for each valid subproblem $SP(i,j,c_1,c_2,c_3)$.
Its arc set $A$ contains the following arcs:
    An arc $(s,SP(1,i,\ast,c_2,c_3))$ for
    each $i<n$ and $c_2,c_3\in[3]$ such that $SP(1,i,\ast,c_2,c_3)$ is valid;
    an arc $(SP(i,n,c_1,c_2,\ast),t)$ for
    each $i>1$ and $c_1,c_2\in[3]$ such that $SP(i,n,c_1,c_2,\ast)$ is valid;
    arcs $(s,SP(1,n,\ast,c_2,\ast))$ and
    $(SP(1,n,\ast,c_2,\ast),t)$ if $SP(1,n,\ast,c_2,\ast)$ is valid;
    an arc $(SP(i,j,c_1,c_2,c_3),SP(j+1,k,c_2,c_3,c_4)$
    for each $i\le j\le k-1$ and each $c_1,c_2,c_3,c_4\in[3]$ (or $c_1=\ast$
    or $c_4=\ast$ if $i=1$ or $k=n$, respectively)
    such that $SP(i,j,c_1,c_2,c_3)$ and $SP(j+1,k,c_2,c_3,c_4)$
    are both valid.
\end{definition}

\begin{theorem}
\label{the:theorem1}%
\label{th:licol-cat}
\licol\ can be solved in polynomial time for
caterpillar-convex bipartite graphs.
\end{theorem}

\begin{proof}
Let a caterpillar-convex bipartite graph $G=(X\cup Y,E)$ with
caterpillar representation $T=(X,F)$ be
given. Let $n$ denote the number of backbone vertices in~$T$,
and $|G|=|X|+|Y|+|E|$ the size of~$G$.
The algorithm, shown in Algorithm~\ref{algo1}, first
computes the solutions to all $O(n^2)$ subproblems $SP(i,j,c_1,c_2,c_3)$.
This can be done in linear time $O(|G|)$ per subproblem
(Lemma~\ref{lem:SP}), and thus in time $O(n^2 |G|)$ overall.

Then it constructs the auxiliary DAG $H$ (Definition~\ref{def:DAG})
and checks if $H$ contains a path $P$ from $s$ to $t$.
As $H$ contains $O(n^2)$ vertices and $O(n^3)$ edges,
the construction of $H$ and the check for the existence
of an $s$-$t$ path can be carried out in $O(n^3)$ time.

Finally, if an $s$-$t$ path $P$ is found, the colorings
corresponding to the subproblems on $P$ can be combined
into a list 3-coloring of $G$ in $O(|G|)$ time.
Thus, the overall running-time of the algorithm
can be bounded by $O(n^2|G|)$, which is polynomial
in the size of the input. If the caterpillar representation
$T=(X,F)$ is not given as part of the input, it can
be computed via our recognition algorithm
(Section~\ref{sec:recognition}) in polynomial time
(the proof of Theorem~\ref{th:recogcat} shows
that the time for computing $T=(X,F)$
is at most $O(|X\cup Y|^3)$).

To show that the algorithm is correct, assume first
that the algorithm finds an $s$-$t$ path $P$ in~$H$.
Let $\mathcal{S}$ be the set of valid subproblems
on~$P$. By construction, each backbone vertex
receives the same color in the at most three
subproblems in $\mathcal{S}$ in which it occurs.
Each leaf of the caterpillar occurs in exactly
one subproblem in $\mathcal{S}$. Every vertex
in $Y$ that occurs in more than one subproblem
in $\mathcal{S}$ must receive the same color
in each such subproblem (because it must
be adjacent to the two backbone vertices
with different colors at the border between
any two consecutive subproblems in which it
is contained). The other vertices in $Y$ occur
in only one subproblem. Hence, the coloring
obtained by the algorithm is a proper list 3-coloring
of~$G$.

For the other direction, assume that $G$ admits
a proper list 3-coloring. Then partition the
backbone $b_1,b_2,\ldots,b_n$ into maximal
segments $B_{1,i_1}, B_{i_1+1,i_2}, \ldots,$
$B_{i_{k-1}+1,i_k}$ for some $k\ge 1$ and
$i_k=n$, so that all backbone vertices in each
segment receive the same color.
Let the color of the backbone vertices
in $B_{i_{j-1}+1,i_j}$ be $c_j$ (where $i_{j-1}=0$ if $j=1$),
for $1\le j\le k$. This implies that the
subproblems $SP(1,i_1,\ast,c_1,c_2)$,
$SP(i_1+1,i_2,c_1,c_2,c_3)$, \ldots,
$SP(i_{k-1}+1,i_k,c_{k-1},c_k,\ast)$
are all valid and constitute an
$s$-$t$-path in the DAG $H$. Therefore,
the algorithm will output a proper list 3-coloring.
\end{proof}

As comb-convex bipartite graphs are a subclass of caterpillar-convex
bipartite graphs, we obtain:

\begin{corollary}\label{cor:licol-cat}
   {\sc Li $3$-col} can be solved in polynomial time for comb-convex bipartite graphs.
\end{corollary}

Combining Corollary \ref{cor:licol-cat}
with Theorem~4 in \cite{brettell_arxiv} and the polynomial-time solvability of {\sc Li $k$-col} for $k\leq 2$ \cite{erdos,vizing} yields a complexity dichotomy:
{\sc Li $k$-col} is polynomial-time solvable on comb-convex bipartite graphs when $k\leq 3$; otherwise, it is NP-complete.

\section{Recognition of Caterpillar-Convex Bipartite Graphs}
\label{sec:section4}\label{sec:recognition}

We give a polynomial-time recognition algorithm for caterpillar-convex bipartite graphs.
We are given a bipartite graph $G=(X\cup Y,E)$ and want to decide whether it is caterpillar-convex and, if so, construct a caterpillar representation $T=(X,F)$.
First, we assume that a specific partition
of the vertex set into independent sets $X$ and $Y$ is given as
part of the input, and we want to decide whether there is
a caterpillar representation $T=(X,F)$ with respect to that given
bipartition (i.e., the vertex set of the caterpillar is the independent
set $X$ that was specified in the input). At the end of this section, we will discuss
how to handle the case that the bipartite graph is given without a specific
bipartition of the vertex set and we want to decide whether the
vertex set can be partitioned into independent sets $X$ and $Y$
in such a way that there is a caterpillar representation with
respect to that bipartition.

The main idea of the algorithm for recognizing caterpillar-convex bipartite graphs
is to construct an auxiliary DAG $D$ on vertex set $X$ in such a way that the sinks in $D$ can be used as the backbone vertices of~$T$.
To make this work, it turns out that we first need to remove
some vertices from $G$ that have no effect on whether $G$ is
caterpillar-convex.
First, we show that we can
remove
isolated vertices from $X$ and vertices of degree
0 or 1 from $Y$.

\begin{lemma}
\label{lem:removex0}
Let $x\in X$ be a vertex with degree $0$, and let
$G'$ be the graph obtained from $G$ by removing~$x$.
Then $G'$ is caterpillar-convex
if and only if $G$ is caterpillar-convex.
Furthermore,
a caterpillar representation of $G$ can be constructed
from a caterpillar representation of $G'$ by adding $x$
in a suitable location.
\end{lemma}

\begin{proof}
Let $T$ be a caterpillar representation of $G$.
If $x$ is a leaf in $T$, we obtain a caterpillar
representation $T'$ of $G'$ simply by removing~$x$ from~$T$.
If $x$ is a backbone vertex in $T$ with at least one leaf
$\ell$ attached to it, we obtain $T'$ by replacing $x$
in the backbone with $\ell$. If $x$ is a backbone vertex
in $T$ without leaves attached to it, we obtain $T'$ by
removing $x$ and making the two former backbone neighbors of
$x$ adjacent (if $x$ had two backbone neighbors).

For the other direction, let $T'$ be a caterpillar
representation of $G'$. We can obtain a caterpillar
representation $T$ of $G$ from $T'$ by adding $x$ as a backbone
vertex to one end of the backbone of $T'$.
\end{proof}

\begin{lemma}
\label{lem:removey01}
Let $y\in Y$ be a vertex with degree $0$ or~$1$, and let
$G'$ be the graph obtained from $G$ by removing~$y$.
Then $G'$ is caterpillar-convex
if and only if $G$ is caterpillar-convex.
Any caterpillar representation of $G'$ is also
a caterpillar representation of~$G$.
\end{lemma}

\begin{proof}
It is clear that any caterpillar representation $T$
of $G$ is also a caterpillar representation of~$G'$.
For the other direction,
let $T'$ be a caterpillar representation of $G'$.
The neighborhood of $y$ induces
an empty subtree or a single-vertex subtree in~$T'$,
so $T'$ is also a caterpillar representation of~$G$.
\end{proof}

We call a pair of vertices $x_i$ and $x_j$ \emph{twins} if $N_G(x_i)=N_G(x_j)$.
The twin relation on $X$ partitions $X$ into equivalence classes,
such that $x_1,x_2\in X$ are twins if and only if they are in the same class.
We say that two twins $x,x'$ are \emph{special twins}
if $\{x,x'\}$ is an equivalence class of the twin relation
on~$X$ and if there is $y\in Y$ with $N_G(y)=\{x,x'\}$.
Now, we show that removing
a twin from $X$ (with some
additional modification in the case of special twins)
has no effect on whether the graph is caterpillar-convex or not.

\begin{lemma}
\label{lem:removetwin}
Let $x,x'\in X$ be twins of non-zero degree, and let $G'=(X'\cup Y',E')$ be the graph obtained
from $G$ by deleting $x$.
If $x,x'$ are special twins in $G$,
then modify $G'$ by adding a new vertex $\bar{x}$
to $X'$, a new vertex $\bar{y}$ to $Y'$, and the edges $\{x',\bar{y}\}$
and $\{\bar{x},\bar{y}\}$ to $E'$. 
Then $G$ is caterpillar-convex
if and only if $G'$ is caterpillar-convex.
Furthermore,
a caterpillar representation of $G$ can be constructed
from a caterpillar representation of $G'$ by adding $x$
in a suitable location (and removing $\bar{x}$ if it
has been added to~$G'$).
\end{lemma}

\begin{proof}
First, consider the case that $x$ and $x'$ are not
special twins.
Assume that $G$ is caterpillar-convex, and
let $T$ be a caterpillar representation.
If at least one of $x$ and $x'$ is a leaf in $T$,
we can assume without loss of generality that $x$ is
a leaf (because the graph obtained from $G$ by deleting
$x$ and the graph obtained by deleting $x'$ are isomorphic,
as $x$ and $x'$ are twins). In that case, removing $x$
from $T$ yields a caterpillar $T'$ that is a caterpillar
representation of $G'$. Now assume that both $x$ and
$x'$ are backbone vertices in~$T$. Form a caterpillar
$T'$ by attaching the leaves in $L(x)$ (where
$L(x)$ denotes the set of leaves attached to backbone
vertex $x$ in caterpillar~$T$) as leaves to~$x'$,
removing $x$, and adding an edge between the two previous
backbone neighbors of $x$ (unless $x$ was an end vertex
of the backbone path). It is easy to see that $T'$ is a
caterpillar representation of $G'$. Hence, in both
cases it follows that $G'$ is caterpillar-convex.

For the other direction, assume that $G'$ is caterpillar-convex,
with caterpillar representation~$T'$. If $x'$ is a backbone
vertex in $T'$, we attach $x$ as leaf vertex to $x'$ in $T'$ to obtain
a caterpillar representation $T$ of~$G$.
If $x'$ is a leaf in $T'$ but $G'$ contains a twin $x''$ of $x'$ that is a backbone
vertex in $T'$, we attach $x$ as leaf vertex to $x''$ in $T'$ to obtain
a caterpillar representation $T$ of~$G$.
It remains to handle the case that $x'$ and all its twins (if any)
in $G'$ are leaf vertices in~$T'$.
If $x'$ has a twin $x''$ in $G'$, this means that
there is no $y\in Y'$ with $N_{G'}(y)=C$, where $C$
is the equivalence class of $x'$ in $X'$; otherwise,
$N_{G'}(y)$ would not be connected in~$T'$.
Therefore, there is also no $y\in Y$ with $N_G(y)=C\cup \{x\}$.
Thus, if we denote by $b$ the backbone vertex in $T'$
to which $x'$ is attached, we must have
$b\in N_{G'}(y)$ for every $y\in Y'$
with $C \subseteq N_{G'}(y)$, and there must
be at least one such $y$ as $x'$ has non-zero degree.
This implies $b\in N_G(y)$ for every $y\in Y$
with $C\cup \{x\}\subseteq N_G(y)$.
Then we can attach $x$ as leaf vertex to that backbone vertex
$b$ in $T'$ to obtain
a caterpillar representation $T$ of $G$.
If $x'$ does not have a twin in $G'$,
we know that $\{x,x'\}$ is an
equivalence class in $X$ and there is no
$y\in Y$ with $N_G(y)=\{x,x'\}$ (otherwise,
$x$ and $x'$ would be special twins).
Thus, if we denote by $b$ the backbone vertex in $T'$
to which $x'$ is attached, we must have
$b\in N_{G'}(y)$ for every $y\in Y'$
with $x'\in N_{G'}(y)$.
This implies $b\in N_G(y)$ for every $y\in Y$
with $\{x,x'\}\subseteq N_G(y)$.
Then we can attach $x$ as leaf vertex to that backbone vertex
$b$ in $T'$ to obtain
a caterpillar representation $T$ of $G$.
Hence, it follows
that $G$ is caterpillar-convex.

Now, we deal with the case that $x$ and $x'$ are special twins.
First, assume that $G$ is caterpillar-convex,
with caterpillar representation~$T$.
As there is $y\in Y$ with $N_G(y)=\{x,x'\}$,
it is not possible that both $x$ and $x'$ are
leaves in~$T$. Without loss of generality,
assume that $x'$ is a backbone vertex.
To obtain $T'$ from $T$, proceed as follows.
First, if $x$ is a leaf, remove $x$ from $T$, and if $x$ is
a backbone vertex,
attach all the leaves in $L(x)$ to $x'$,
remove $x$, and make the previous backbone
neighbors of $x$ adjacent to each other
(only in case $x$ was not an end vertex
of the backbone). Then,
add $\bar{x}$ as a leaf attached to~$x'$.
Observe that $N_{G'}(\bar{y})=\{x',\bar{x}\}$ induces
a connected subgraph of~$T'$.
Therefore, $T'$ is a caterpillar representation
of $G'$, and so $G'$ is caterpillar-convex.

For the other direction, assume that $G'$ is caterpillar-convex,
with caterpillar representation~$T'$.
If $x'$ is a leaf vertex in $T'$, then it must be
attached to the backbone vertex $\bar{x}$ (as $N_{G'}(\bar{y})=\{x',\bar{x}\}$).
Furthermore, the only vertex in $Y'$ that is adjacent to $\bar{x}$
is $\bar{y}$. Therefore, we can swap the positions of $x'$ and
$\bar{x}$ in $T'$, and the resulting tree is still a caterpillar representation
of~$G'$. Hence, we can assume that $x'$ is a backbone vertex in $T'$.
To obtain a caterpillar representation $T$ of~$G$, we add
$x$ as a leaf attached to~$x'$ to~$T'$, and we remove $\bar{x}$: If
$\bar{x}$ is a leaf vertex, we simply remove it, and if $\bar{x}$
is a backbone vertex, we attach all the leaves in $L(\bar{x})$
to an arbitrary other backbone vertex (for example, to~$x'$),
remove $\bar{x}$, and make the two previous backbone neighbors
of $\bar{x}$ adjacent to each other (unless $\bar{x}$ was an
end vertex of the backbone). The latter operation is correct
as the only vertex in $Y'$ that is adjacent to $\bar{x}$ is~$\bar{y}$.
The resulting tree is a caterpillar representation of~$G$.
\end{proof}

We remark that the special treatment of special twins in Lemma~\ref{lem:removetwin}
is necessary because there is a graph $G=(X\cup Y,E)$ with special twins that does not
have a caterpillar representation $T=(X,F)$, while simply removing one of the two special twins (without
adding the extra vertices $\bar{x}$ and $\bar{y}$) would produce
a graph $G'=(X'\cup Y',E')$ that has a caterpillar representation $T'=(X',F')$.
An example of such a graph
is the graph with $X=\{a,b,c,f,g,x,x'\}$ where the neighborhoods
of the vertices in $Y$ are $\{a,f\}, \{a,b\}, \{b,x,x'\}, \{x,x'\}, \{b,c\}, \{c,g\}$.
Here, the vertices $x$ and $x'$ are special twins, and the graph obtained
after removing $x$ has the caterpillar representation with backbone
path $abc$ and leaf~$f$ attached to~$a$, leaf $x'$ attached to~$b$, and leaf
$g$ attached to~$c$.

Let $G_1=(X_1\cup Y_1,E_1)$ be the graph obtained from $G=(X\cup Y,E)$ by
removing vertices of degree $0$ from~$X$, vertices
of degree $0$ or $1$ from~$Y$,
and twins from $X$ (with the extra modification detailed in Lemma~\ref{lem:removetwin}
in case of special twins) as long as such vertices exist.
Lemmas~\ref{lem:removex0}--\ref{lem:removetwin}
imply:

\begin{corollary}
\label{cor:removevertices}%
$G_1$ is caterpillar-convex if and only if $G$ is caterpillar-convex.
\end{corollary}

We now define a directed graph $D=(X_1,A)$
based on $G_1$: For every pair of distinct vertices $x,x'\in X_1$,
we let $D$~contain the arc $(x,x')$ if and only if $N_{G_1}(x)\subseteq N_{G_1}(x')$, i.e.,
we add the arc $(x,x')$ if and only if every vertex in $y$ that is adjacent to $x$ in $G_1$ is also adjacent to $x'$ in~$G_1$. Note that
$D$ is transitive: If it contains two arcs $(x,x')$ and $(x',x'')$,
it must also contain $(x,x'')$.

\begin{lemma}
\label{lem:acyclic}%
    $D$ is a directed acyclic graph.
\end{lemma}

\begin{proof}
Assume there is a cycle on vertices $x_i, x_{i+1}, \ldots, x_j$ in $D$. Then, $N(x_i)\subseteq N(x_{i+1}) \subseteq \cdots \subseteq N(x_j)\subseteq N(x_i)$. Thus $N(x_i)=N(x_{i+1})= \cdots = N(x_j)$, and so $x_i, x_{i+1}, \ldots, x_j$ are twins, a contradiction because there are no twins in $X_1$.
Thus $D$ is acyclic.
\end{proof}

\begin{lemma}
\label{lem:nobbarc}%
If $G_1=(X_1\cup Y_1, E_1)$ is caterpillar-convex, there is a caterpillar representation $T_1=(X_1, F)$ in which no two backbone vertices are connected by an arc in~$D$.
\end{lemma}

\begin{proof}
Let $T_1=(X_1, F)$ be a caterpillar representation for $G_1$ in which there are two backbone vertices $x_i$ and $x_j$ that are connected
by an arc $(x_i,x_j)$ in $D$. If $x_i$ and $x_j$
are not adjacent on the backbone path $P_1$ of $T_1$, observe that
$x_i$ also has an arc to every vertex between $x_i$ and
$x_j$ on $P_1$. This is because, in $G_1$, each neighbor of $x_i$ is also
adjacent to $x_j$ and hence to all vertices between $x_i$ and
$x_j$ on $P_1$. Thus, we can choose $x_i$ and $x_j$ to be
adjacent backbone vertices that have an arc $(x_i,x_j)$ in~$D$.

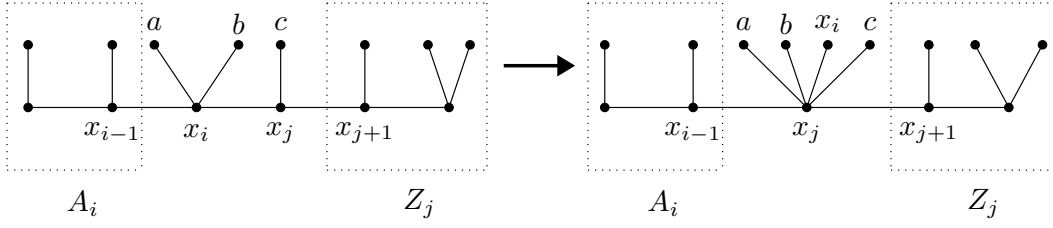
\begin{figure}[t] 
\begin{center}
\resizebox{0.99\textwidth}{!}{%
\begin{tikzpicture}[dot/.style={draw,circle,minimum size=1mm,inner sep=0pt,outer sep=0pt}]

\draw[dotted] (-0.5,4.25) -- (1.1,4.25);
\draw[dotted] (-0.5,4.25) -- (-0.5,2.25);
\draw[dotted] (1.1,4.25) -- (1.1,2.25);
\draw[dotted] (-0.5,2.25) -- (1.1,2.25);

\draw (0.4,1.85) node {$A_i$};

\draw[dotted] (3.3,4.25) -- (5.2,4.25);
\draw[dotted] (3.3,4.25) -- (3.3,2.25);
\draw[dotted] (5.2,4.25) -- (5.2,2.25);
\draw[dotted] (3.3,2.25) -- (5.2,2.25);

\draw (4.4,1.85) node {$Z_j$};

\draw[dotted] (6.4,4.25) -- (8,4.25);
\draw[dotted] (6.4,4.25) -- (6.4,2.25);
\draw[dotted] (8,4.25)   -- (8,2.25);
\draw[dotted] (6.4,2.25) -- (8,2.25);

\draw (7.3,1.85) node {$A_i$};

\draw[dotted] (10,4.25) -- (12,4.25);
\draw[dotted] (10,4.25) -- (10,2.25);
\draw[dotted] (12,4.25)  -- (12,2.25);
\draw[dotted] (10,2.25) -- (12,2.25);

\draw (11.1,1.85) node {$Z_j$};

\draw[-Triangle, very thick](5.4,3.5) -- (6.25,3.5);

\draw (1.25,4) node {$a$};
\draw (2.25,4) node {$b$};
\draw (2.75,4) node {$c$};

\draw (0.75,2.7) node {$x_{i-1}$};
\draw (1.75,2.7) node {$x_i$};
\draw (2.75,2.7) node {$x_j$};
\draw (3.75,2.7) node {$x_{j+1}$};

\draw (8.25,4) node {$a$};
\draw (8.75,4) node {$b$};
\draw (9.25,4) node {$x_i$};
\draw (9.75,4) node {$c$};

\draw (7.65,2.7) node {$x_{i-1}$};
\draw (9,2.7) node {$x_j$};
\draw (10.45,2.7) node {$x_{j+1}$};

\coordinate [dot,fill = black] (X1) at (-0.25,3);
\coordinate [dot,fill = black] (X2) at (0.75,3);
\coordinate [dot,fill = black] (X3) at (1.75,3);
\coordinate  [dot,fill = black](X4) at (2.75,3);
\coordinate  [dot,fill = black](X5) at (3.75,3);
\coordinate  [dot,fill = black](X6) at (4.75,3);

\coordinate  [dot,fill = black](X8) at (6.6,3);
\coordinate  [dot,fill = black](X9) at (7.65,3);
\coordinate  [dot,fill = black](X10) at (9,3);
\coordinate  [dot,fill = black](X11) at (10.45,3);
\coordinate  [dot,fill = black](X12) at (11.4,3);

\coordinate[dot,fill = black] (Y1) at (-0.25,3.75);
\coordinate [dot,fill = black] (Y2) at (0.75,3.75);
\coordinate [dot,fill = black] (Y3) at (1.25,3.75);
\coordinate [dot,fill = black] (Y4) at (2.25,3.75);
\coordinate [dot,fill = black] (Y5) at (2.75,3.75);
\coordinate [dot,fill = black] (Y6) at (3.75,3.75);
\coordinate [dot,fill = black] (Y7) at (4.5,3.75);
\coordinate [dot,fill = black] (Y8) at (5,3.75);

\coordinate [dot,fill = black] (Y9) at (6.6,3.75);
\coordinate [dot,fill = black] (Y10) at (7.65,3.75);
\coordinate [dot,fill = black] (Y11) at (8.25,3.75);
\coordinate [dot,fill = black] (Y12) at (8.75,3.75);
\coordinate [dot,fill = black] (Y13) at (9.25,3.75);
\coordinate [dot,fill = black] (Y14) at (9.75,3.75);
\coordinate [dot,fill = black] (Y15) at (10.45,3.75);
\coordinate [dot,fill = black] (Y16) at (11,3.75);
\coordinate [dot,fill = black] (Y17) at (11.8,3.75);

\draw [black] (X1) -- (Y1);
\draw [black] (X2) -- (Y2);
\draw [black] (X3) -- (Y3);
\draw [black] (X3) -- (Y4);
\draw [black] (X4) -- (Y5);
\draw [black] (X5) -- (Y6);
\draw [black] (X6) -- (Y7);
\draw [black] (X6) -- (Y8);
\draw [black] (X1) -- (X2);
\draw [black] (X2) -- (X3);
\draw [black] (X3) -- (X4);
\draw [black] (X4) -- (X5);
\draw [black] (X5) -- (X6);

\draw [black] (X8) -- (Y9);
\draw [black] (X9) -- (Y10);
\draw [black] (X10) -- (Y11);
\draw [black] (X10) -- (Y12);
\draw [black] (X10) -- (Y13);
\draw [black] (X10) -- (Y14);
\draw [black] (X11) -- (Y15);
\draw [black] (X12) -- (Y16);
\draw [black] (X12) -- (Y17);

\draw [black] (X8) -- (X9);
\draw [black] (X9) -- (X10);
\draw [black] (X10) -- (X11);
\draw [black] (X11) -- (X12);

\end{tikzpicture}}
\end{center}
 \caption{Caterpillar transformation from $T_1$ (left) to $T_1'$ (right) used in the proof of Lemma~\ref{lem:nobbarc}.} \label{fig:transformation}
\end{figure}%
Let $L(b)$ denote the set of leaf vertices
attached to a backbone vertex $b$ in $T_1$.
Observe that each leaf $\ell\in L(b)$ must have an arc to $b$
in~$D$. This is because each neighbor (in $G_1$) of $\ell$ has degree at
least~$2$ and is hence also adjacent to~$b$.
Therefore, every leaf in $L(x_i)$ has an arc to $x_i$ and,
by transitivity of $D$, also an arc to~$x_j$.
We now create a new caterpillar $T_1'$ from $T_1$ by
(1) attaching $x_i$ and all the leaves in $L(x_i)$ as leaves to $x_j$,
and (2) making the two previous backbone neighbors of $x_i$ adjacent (only
if $x_i$ was not an end vertex of the backbone path).
See Fig.~\ref{fig:transformation} for an illustration.

We will show in the remainder of this proof that $T_1'$ is
a caterpillar representation of~$G_1$. By applying the same
operation repeatedly as long as there exist two backbone vertices
that are connected by an arc in~$D$, the statement
of the lemma follows.

Assume the vertices on the backbone path of $T_1$
are $x_1,x_2,\ldots,x_i,x_j,\ldots,x_r$ with $j=i+1$.
Define $A_i=\bigcup_{k=1}^{i-1}(\{x_k\}\cup L(x_k))$
and $Z_j=\bigcup_{k=j+1}^r(\{x_k\}\cup L(k_k))$.
These are the parts of the backbone that are not
affected by the transformation from $T_1$ to $T_1'$.
Note that $A_i$ and/or $Z_j$ can also be empty.
We now prove for every $y\in Y$
that $N(y)$ induces a tree in~$T_1'$:

\noindent\textbf{Case 1:} $N(y) \cap (\{x_i, x_j\} \cup L(x_i) \cup L(x_j)) = \emptyset$.
Since $N(y)$ induces a tree in $T_1$, we must have $N(y)\subseteq A_i$
or $N(y)\subseteq Z_j$, and hence $N(y)$ also induces a tree in~$T_1'$. 

\noindent\textbf{Case 2:} $N(y) \cap (\{x_i, x_j\} \cup L(x_i) \cup L(x_j)) \neq \emptyset$. Note that $x_j\in N(y)$ as all the vertices
in $\{x_i\}\cup L(x_i)\cup L(x_j)$ have an arc to $x_j$
in $D$.
As $N(y)$ induces a tree in $T_1$,
we observe that
    $N(y) \cap A_i$ is either empty or contains $x_{i-1}$,
and that $N(y) \cap Z_j$ is either empty or contains $x_{j+1}$.
As the caterpillar part on $A_i$ and $Z_i$ has not changed
in the transformation from $T_1$ to $T_1'$, $N(y)\cap A_i$
is either empty or induces a tree containing $x_{i-1}$
in $T_1'$, and $N(y)\cap Z_j$ is either empty or induces
a tree containing $x_{j+1}$ in $T_1'$.
Furthermore, $N(y)\setminus(A_i\cup Z_j)$ contains
$x_j$ and some subset of the leaf neighbors of $x_j$
in $T_1'$ and hence induces a star containing~$x_j$ in~$T_1'$.
As $x_j$ is adjacent to $x_{i-1}$ and $x_{j+1}$ (if those
vertices exist),
$N(y)$~induces a tree
in~$T_1'$.
\end{proof}

\begin{lemma}
\label{lem:sinksbb}%
If $G_1=(X_1\cup Y_1, E_1)$ is caterpillar-convex, there is a caterpillar representation $T_1=(X_1, F)$ such that the set
of backbone vertices is exactly the set of sinks in~$D$.
\end{lemma}

\begin{proof}
By Lemma~\ref{lem:nobbarc}, there exists a caterpillar representation
$T_1$ of $G_1$ in which no two backbone vertices are connected
by an arc in~$D$. Furthermore, every leaf attached to a backbone
vertex (in $T_1$) has an arc (in $D$) to that backbone vertex (because
every $y\in Y$ has degree at least~$2$). A backbone
vertex cannot have an arc (in $D$) to a leaf attached to it (in $T_1$), as $D$ is acyclic (Lemma~\ref{lem:acyclic}). Finally,
a backbone vertex $b$ cannot have an arc (in $D$) to
a leaf vertex $\ell$ attached to a different backbone vertex $b'$
because that would imply that $b$ has an arc to $b'$
(since every vertex in $y$ that is adjacent to $b$ is also
adjacent to $\ell$ and hence, as $N(y)$ induces a tree in $T_1$,
also to $b'$). Therefore, the backbone vertices of $T_1$ are
exactly the sinks (vertices without outgoing edges) of~$D$.
\end{proof}

\begin{algorithm}[tbp]
\caption{Recognition Algorithm for Caterpillar-Convex Bipartite Graphs}
\label{algo2}
\begin{algorithmic}[1]
\Require A bipartite graph $G=(X \cup Y,E)$
\Ensure Either return a caterpillar representation $T=(X,F)$ of $G$, or decide that $G$ is not caterpillar-convex and return `fail'
\State Obtain $G_1=(X_1\cup Y_1,E_1)$ from $G$ by removing
vertices of degree~$0$ from $X$, vertices of degree $0$ or $1$ from $Y$,
and twins from $X$ (with the extra modification stated in Lemma~\ref{lem:removetwin}
in case of special twins),
as long as any such vertex exists (Lemmas~\ref{lem:removex0}--\ref{lem:removetwin})
\State Create a directed graph $D=(X_1, A)$ that
contains the arc $(x,x')$ if and only if
$N_{G_1}(x)\subseteq N_{G_1}(x')$
\State $B=$ the set of sinks in $D$, $L=$ all other vertices in $D$
\State Use an algorithm for consecutive ones \cite{hsu} to order $B$.
If the algorithm fails, return `fail'.
\State Form caterpillar $T_1$ by taking the ordered backbone $B$ and attaching each vertex in $L$ as leaf to an arbitrary vertex in $B$ to which it has an arc in $D$
\State Obtain $T$ from $T_1$ by adding the vertices that were deleted from $X$ in Step 1 (and removing vertices that have been added when
special twins were processed) (Lemmas~\ref{lem:removex0} and~\ref{lem:removetwin}).
\State Return $T$
\end{algorithmic}
\end{algorithm}

\begin{theorem}
\label{theorem:the3}\label{th:recogcat}%
Algorithm \ref{algo2} decides in polynomial time whether a given bipartite
graph $G=(X\cup Y,E)$ is caterpillar-convex and, if so, outputs a
caterpillar representation $T=(X,F)$.
\end{theorem}

\begin{proof}
Let $G=(X \cup  Y, E)$ be a bipartite graph. Let $n=|X\cup Y|$ and $m=|E|$.
First, the algorithm
removes vertices of degree $0$ from~$X$, vertices
of degree $0$ or $1$ from~$Y$,
and twins from $X$ (with the extra modification detailed in Lemma~\ref{lem:removetwin}
in case of special twins) as long as such vertices exist.
The resulting
graph is $G_1=(X_1\cup Y_1,E_1)$. By Corollary~\ref{cor:removevertices},
$G_1$ is caterpillar-convex if and
only if $G$ is caterpillar-convex.

We show that $G_1$ can be computed from $G$
in $O(n^2)$ time. If $G$ is given as adjacency matrix,
we can compute an adjacency list representation in $O(n^2)$
time. Then, in $O(n+m)$ time, we can compute the
degree of every vertex and make a list $L_0$ of vertices in $X$
with degree~$0$ and vertices in $Y$ of degree~$0$
or~$1$. As long as $L_0$ is non-empty, we remove
a  vertex $v$ from $L_0$, decrease the degree of its
neighbor (if any) by~1, and delete $v$ from~$G$.
If $v$ had a neighbor (which is only possible
if $v\in Y$) and that neighbor now has degree~$0$,
we add that neighbor (which must be in $X$) to~$L_0$.
This takes $O(n)$ time as each vertex removed from $L_0$
can be processed in $O(1)$ time and we remove at most $n$ vertices.
Let $X'$ and $Y'$ denote the vertices that have not yet been deleted
at this stage.
Next, we compute a partition of $X'$ into equivalence classes,
where $x$ and $x'$ are in the same equivalence class if and
only if they are twins in the current graph. This can be
done in $O(n^2)$ time.
For example, one can start with a partition $\mathcal{P}$
consisting of a single equivalence class equal to $X'$ and then, for each $y\in Y'$, refine $\mathcal{P}$ in $O(n)$ time so that each equivalence class
$C$ with $N(y)\cap C\neq \emptyset$ and $C\setminus N(y)\neq \emptyset$
gets split into $N(y)\cap C$ and $C\setminus N(y)$.
During that process, we can also determine for each
resulting equivalence class $C$ whether there is
a vertex $y\in Y'$ with $N(y)=C$, without exceeding
the time bound of $O(n^2)$.
Then, for each equivalence class $C$ of size at least $2$,
we remove all but one vertex in $C$ from $X'$ and update the
degrees of the neighbors of the deleted vertices accordingly;
the vertices in $Y'$ whose degree becomes $1$ in this process
are again added to the list $L_0$ and then deleted from the graph in the same
way as above.
Besides, if $x'$ is the vertex in $C$ that we
do not remove and if there is a $y\in Y'$ with $N(y)=C$,
we add new vertices $\bar{x},\bar{y}$ and edges $\{\bar{x},\bar{y}\}$
and $\{x',\bar{y}\}$ to the graph in order to implement the
treatment of special twins according to Lemma~\ref{lem:removetwin}.
We note that processing one equivalence class $C$
in this way does not alter the other equivalence classes.
In particular, it cannot happen that two other equivalence classes $C_1$
and $C_2$ get merged because the processing of $C$ leads
to the deletion of some vertices with degree~1 in $Y'$:
the vertices deleted from $Y'$ while processing $C$ are not
adjacent to any vertex outside $C$, and hence their deletion
has no effect on the twin relationship between vertices outside~$C$.
The time for processing one equivalence class $C$ can
be bounded by the sum of the degrees of the vertices in $C$,
and hence the time for processing all equivalence classes
is bounded by $O(m)\subseteq O(n^2)$.
Thus, $G_1$ can be obtained in $O(n^2)$ time.
Note that $G_1$ has $O(n)$ vertices as the number of new vertices
added to $G_1$ is bounded by $2s$, where $s\le \frac{n}{2}$ is the number of
equivalence classes $C$ with $|C|\ge 2$
for which a vertex $y\in Y'$ with $N(y)=C$ exists.

Next, the algorithm constructs the directed graph $D=(X_1,A)$ from $G_1=(X_1\cup Y_1, E_1)$ by adding an arc $(x_i,x_j)$ for $x_i,x_j\in X_1$ if $N_{G_1}(x_i)\subseteq N_{G_1}(x_j)$. For any two vertices $x_i,x_j\in X_1$
one can trivially check in $O(n)$ time whether $N(x_i)\subseteq N(x_j)$,
so $D$ can easily be constructed in $O(n^3)$ time. As
$|X'|=O(n)$ and $|A|=O(n^2)$,
the set $B$ of sinks and the
set $L$ of remaining vertices can be determined in $O(n^2)$ time
once $D$ has been constructed.

Once the set $B$ has been determined, we create a set system
$\mathcal{S}$ containing for every $y\in Y$ the set $N(y)\cap B$
and apply an algorithm for checking the consecutive ones
property~\cite{hsu} to check if $B$ can be ordered in such
a way that every set in $\mathcal{S}$ consists of consecutive
vertices. If so, the resulting order is used to determine
the order in which $B$ forms the backbone path. Otherwise,
$G_1$ (and hence $G$) cannot be caterpillar-convex (cf.~Lemma~\ref{lem:sinksbb}), and
the algorithm returns `fail'. As the input to the consecutive
ones algorithm can be represented as a matrix of size $O(n^2)$,
running Hsu's linear-time algorithm~\cite{hsu} on $\mathcal{S}$
takes $O(n^2)$ time.

Next, the algorithm attaches each vertex $\ell\in L$
as a leaf to an arbitrary vertex $b\in B$ to which it has
an arc in~$D$. It can be shown as follows that every $\ell\in L$
must indeed have at least one arc to a vertex in $B$:
As $D$ is acyclic, every vertex $\ell$ that is not a sink
must have a directed path leading to some sink~$b$, and
as $D$ is transitive, the arc $(\ell,b)$ must exist.
Attaching $\ell$ to $b$ yields a valid caterpillar
representation for the following reason:
As every neighbor $y$ of $\ell$ is
also adjacent to $b$, and as $N(y)\cap B$ is a contiguous
segment of $B$, it is clear that $N(y)$ induces a tree in
the resulting caterpillar~$T_1$. Hence, $T_1$ is a caterpillar
representation of~$G_1$. Attaching the vertices of $L$
as leaves to suitable vertices in $B$ can easily be done
in $O(n^2)$ time.

Finally, the vertices that have been deleted in the first
step are added back (and vertices that have been added when
special twins were processed are removed) in order to extend
the caterpillar $T_1$ to a caterpillar representation $T$ of $G$.
This can easily be done in $O(n)$ time per vertex following the
arguments used in the proofs of Lemmas~\ref{lem:removex0}
and~\ref{lem:removetwin}.
By Corollary~\ref{cor:removevertices}, $T$~is a caterpillar representation
of~$G$.

The running-time of the algorithm is dominated by
the time for constructing $D$, which we have bounded
by $O(n^3)$. We remark that we have not attempted to
optimize the running-time, as our main goal was to
show that caterpillar-convex bipartite graphs can
be recognized in polynomial time.
\end{proof}

Finally, we discuss the case that the bipartition of the vertex
set $V$ of the input graph $G=(V,E)$ into independent sets $X$ and $Y$ is not provided as
part of the input.
First, if $G=(V,E)$ is a connected bipartite graph, note that there is a unique
partition of $V$ into two independent sets $Q$ and $R$. We can then run
the recognition algorithm twice, once with $X=Q$ and $Y=R$ and once with
$X=R$ and $Y=Q$. $G$~is caterpillar-convex if and only if at least one of
the two runs of the algorithm produces a caterpillar representation.
If $G=(V,E)$ is not connected, let $H_1,\ldots,H_r$ for some $r>1$ be its connected
components. As just discussed, we can check in polynomial time whether
each connected component $H_j$, $1\le j\le r$, is a caterpillar-convex bipartite graph.
If all $r$ connected components are caterpillar-convex, the whole
graph $G$ is caterpillar-convex, and a caterpillar representation can
be obtained by concatenating the backbones of the caterpillar representations
of the connected components in arbitrary order. If at least one of the
connected components, say, the component~$H_j$, is not caterpillar-convex,
then $G$ is not caterpillar-convex either. This can be seen as follows:
Assume for a contradiction that $G$ is caterpillar-convex while
$H_j$ is not caterpillar-convex. Then let
$T=(X,F)$ be a caterpillar representation of~$G$. Observe that the
subgraph of $T$ induced by~$V(H_j)\cap X$, where $V(H_j)$ denotes
the vertex set of~$H_j$, must be connected. Therefore,
that subgraph of $T$ provides a caterpillar representation of~$H_j$,
a contradiction to our assumption.
This establishes the following corollary.

\begin{corollary}
\label{cor:recogcat}%
There is a polynomial-time algorithm that decides whether a given
bipartite graph $G=(V,E)$ is caterpillar-convex, i.e., whether it
admits a bipartition of $V$ into independent sets $X$ and $Y$ such that
there is a caterpillar representation $T=(X,F)$.
\end{corollary}

\section{Conclusion}
\label{sec:Conclusion}

Determining the computational complexity of {\sc Li $k$-col} for $k\geq 3$ when restricted to comb-convex bipartite graphs was stated as an open problem by Bonomo-Braberman et al.~\cite{flavia}. Subsequently,
the same authors proved that the problem is NP-complete
for $k\geq 4$~\cite{brettell_arxiv}, but the complexity
for $k=3$ was still left open.
In this paper, we resolve this question by showing that {\sc Li $3$-col} is solvable in polynomial time even for caterpillar-convex bipartite graphs, a superclass of comb-convex bipartite graphs. 

Recall that if mim-width is bounded for a graph class $\mathcal{G}$, then {\sc Li $k$-col} is polynomially solvable when it is restricted to $\mathcal{G}$. Polynomial-time solvability of {\sc Li $k$-col} on circular convex graphs is shown by demonstrating that mim-width is bounded for this graph class ~\cite{flavia}. On the other hand, there are graph classes for which {\sc Li 3-col} is tractable but mim-width is unbounded, such as star-convex bipartite graphs~\cite{brettell_arxiv}. By combining our result with Theorem~3 in~\cite{flavia}, we conclude that caterpillar-convex bipartite graphs and comb-convex bipartite graphs also belong to this type of graph classes.
On a much larger graph class, chordal bipartite graphs, the computational complexity of {\sc Li $3$-col} is still open~\cite{huang}.

Finally, as for future work, it would be interesting to see whether one can modify and extend Algorithm~\ref{algo2} to recognize comb-convex bipartite graphs.

\bibliography{caterpillar_arxiv_revised}

\end{document}